\documentclass[12pt,a4paper]{article}
\usepackage[latin1]{inputenc}
\usepackage[T1]{fontenc}
\usepackage{amsmath}
\usepackage{amsfonts}
\usepackage{graphicx}
\usepackage{enumerate}
\usepackage{amssymb,amsthm, mathtools}
\usepackage{tikz}
\usepackage{authblk}
\usepackage{natbib}
\usepackage{hyperref}
\newtheorem{definition}{Definition}
\newtheorem{theorem}{Theorem}
\newtheorem{proposition}{Proposition}
\newtheorem{remark}{Remark}
\newtheorem{lemma}{Lemma}
\newtheorem{example}{Example}

\author{R. Soeiro}
\author{A. A. Pinto}
\title{Group network effects in price competition}
\affil{Universidade do Porto \& Inesc Tec}
\begin{document}
	\maketitle
	\begin{abstract}
		The partition of society into groups, polarization, and social networks are part of most conversations today. How do they influence price competition? We discuss Bertrand duopoly equilibria with demand subject to network effects. Contrary to models where network effects depend on one aggregate variable (demand for each choice), partitioning the dependence into groups creates a wealth of pure price equilibria with profit for both price setters, even if positive network effects are the dominant element of the game. If there is some asymmetry in how groups interact, two groups are sufficient. If network effects are based on undirected and unweighted graphs, at least five groups are required but, without other differentiation, outcomes are symmetric.
	\end{abstract}
	\section{Introduction}
		Consider the class of two-stage games where, in the first stage, a subset of players (e.g., firms, sellers, platforms, etc.) set a price for the service or good they provide; in the second stage, a different set of players (e.g., consumers, buyers, users, etc.) decide on a (priced) option from the first stage while subject to network effects. The equilibrium concept is that of pure price subgame-perfect equilibrium (SPE). Examples of services range from restaurants to online platforms; of goods, from clothing to technology. Network effects, i.e., how second-stage players influence each other's choices, can be positive, negative, or both. 

		In general, first stage price competition is fiercer under positive network effects, which tend to produce corner solutions (at least one firm with no profit), while negative network effects relax competition and allow for interior price solutions (where all firms have profit), \citep*[see][]{Grilo:2001:PCW}. One of the goals of this work is to show that while this idea holds when network effects (for each choice) depend on a single aggregate variable (number of players on that choice), it is not necessarily true when network effects depend on more than one aggregate variable. That is, when there are two or more second-stage groups that influence each other differently, the opposite may happen. The partition into groups might be as refined as small cliques, less detailed as in peer or social groups, or even broader as in democrats/ republicans, smokers/ non-smokers, etc. Interestingly, two aggregate variables (groups) are enough to allow reverting the one dimension result.\footnote{With uniform pricing, players in the first stage pricing game always look at only one aggregate variable: demand.}

		A second goal of the work is a general mathematical treatment of network effects as a determinant of equilibria, in particular, local pure price subgame-perfect equilibria where both firms have profit (SPE+). The general problem with finding these equilibria under positive network effects is that the possibility of a bandwagon response in the second stage creates price deviation incentives in the first stage. So, typically, positive network effects are assumed sufficiently small relative to, for example, a negative network effect (e.g., congestion), a product differentiation effect (e.g., Hotelling), or some lock-in variable (either in the first or second stage). For instance, when there are no network effects, \citet*{Caplin:1991:AAI} show that restrictions on consumer preferences are sufficient to ensure a pure price equilibrium. Take that solution and perturb it with (positive) network effects. As long as the latter are sufficiently small, one can study their impact without disrupting the existing solution.

		Our approach is, in some sense, the inverse. We consider a duopoly and start with a general network effects function, which can depend on several aggregate variables (e.g., groups of consumers). Sufficient conditions for a SPE+ outcome lead to (marginal) network properties on network effects. One can then perturb this solution by introducing different variables (e.g., (intrinsic) product differentiation) but without the need to assume a particular form (say $\rho$-concavity and convex support). In particular, this can help assess the appropriate function and dimension to use when modeling: how many groups are necessary and how do they influence each other? 

		Let's look at a simple stylized example with both positive and negative network effects. Say there are two choices (e.g., restaurants) and two groups, snobs and conformists. The former's utility depends negatively on the number of second-stage players on the same choice, and the latter positively. \citet*{Amaldoss:2005:POC}, for the linear case, use two commonly coupled restrictions to guarantee a SPE+. An assumption on the distributions of preferences of both groups and a bound on conformism. Suppose we change the game by removing both restrictions but allowing conformists to distinguish between snobs and other conformists (network effects, for conformists, now depend on two aggregate variables). If the impact of snobs on conformists is different than that of other conformists\footnote{One typical example is to take snobs as being celebrities having a higher impact on followers than do other followers.}, a SPE+ where both groups split (i.e., its members make different choices) becomes possible. Furthermore, the resulting outcome does not depend on how conformists distinguish others, the relative value of snobbism and conformism, or the sizes of each group.\footnote{The last observation refers to the linear case in question and is not necessarily completely generalizable.}

		Increasing the dimension can open the possibility for SPE+ outcomes, but it can also have the opposite effect, disrupting existing equilibria. For example, when there are only negative network effects. However, the relevant part for constructing SPE+ is the interaction among groups that split (the set of indifferent consumers), so the possibilities and complexity rise with the power set of the partition into groups (which is only a problem if one wishes to show non-existence or uniqueness). That means the type of equilibria where some groups split, while others don't, is always possible and works as going back to a lower dimension (although it produces different outcomes). Thus, increasing the detail in partitioning second-stage players into groups, in general, creates new equilibria.

		Dimension is only one part of the problem. Let's take the example of two-sided markets from \citet*{Armstrong:2006:CIT} (the specific case with uniform pricing), with again two groups but now only positive network effects. Two groups interact via two platforms, like heterosexual dating agencies, nightclubs, credit cards (consumers/ retailers), shopping malls, television channels, and, we will add, online platforms (users/ providers, e.g., lodging, transport, food delivery).\footnote{This serves as an example of how the relative impact of groups on each other produces equilibria, even if some of the mentioned cases would be more natural for the price discrimination literature.}  Each group's decision depends on the number of members of the other group that uses each platform. Armstrong characterizes a SPE+ for the case when network effects are sufficiently small relative to product differentiation (transport costs) using a Hotelling specification. Without differentiation (no transport costs), such SPE+ is not possible. Suppose now that we change the game and introduce intra-group influences. For example, consumers may value how many other consumers a platform attracts (e.g., as a form of quality or reliability signaling).\footnote{In Armstrong's example of dating agencies and nightclubs, the network effects structure modification means to consider them bisexual. Then, a SPE+ is possible if same-sex encounters have a higher value than different-sex ones.} We show that, in this case, if intra-group influences are sufficiently strong, a SPE+ where both groups split among the two platforms becomes possible (meaning it does not require product differentiation), which might seem counterintuitive.\footnote{What might strike as counterintuitive is that introducing intra-group influences produces an equilibrium where groups split. However, second-stage Nash is not the issue. Intra-group influences stabilize first-stage price competition.}

		As the number of considered groups increases, so does the complexity of possible equilibria as it depends on the relative interactions among the different groups. Some authors (e.g., \citet*{Banerji:2009:LNE}) consider a given a priori undirected and unweighted graph, which characterizes how groups are connected. It is an instance of a multilinear network function, with the advantage of reducing the discussion to an adjacency matrix. Because interactions weights are homogeneous and symmetric (either 1 or 0 for connection), properties and the existence of SPE+ outcomes are dependent only on the underlying structure, not the relative weights of network effects. In particular, on the induced subgraph among groups that split. However, homogeneity in interactions makes it harder to allow SPE+.  Even if groups have different sizes, these do not play the same role as a network effect weight (we will see that in the more general multilinear case), forcing an increase in dimension (more groups) to produce a solution. When a SPE+ outcome exists, it is unique, completely symmetric, and independent of the network structure and group sizes. The latter holds for every multilinear model where network effects are symmetric. Thus, weights and asymmetry in connections (directed networks) play a higher role than the underlying structure.

		The technical difficulty in studying SPE under network effects is the multiplicity of second-stage equilibria. The definition involves a second-stage contingent plan, so this is sometimes called a second-stage coordination problem. In a certain sense, we will argue, it is an artificial problem (in this class of games). Our objective here is to characterize local equilibrium outcomes, not predict play, hence multiplicity of game equilibria is not an issue. It is a first-stage maximization difficulty: evaluating an action (pricing) depends on whether deviations from it are profitable, but how to evaluate price deviations when there are multiple second-stage equilibria? In the classic definition, second-stage players (e.g., consumers) would have to announce their choices for every possible pair of prices. However, can you think of a practical example of such a situation? To deal with the issue, we use a slightly different approach, resorting to the concept of the imagined demand curve proposed by \citet*{Kaldor:1934:MRS}:

		\begin{quote}
			The traditional "market demand curve" for a certain product is not the same sort of thing as the demand curve which is relevant in determining the actions of the individual producer. The first denotes a functional relationship between the price and the amounts bought from a particular producer. The second concerns the image of this functional relationship as it exists in the mind of the entrepreneur. 
	
			(...)
	
			The "imagined demand curve" (...), becomes determinate as soon as it exists in the producer's imagination - and since something always must exist there the question of indeterminateness simply does not arise in this case. If on the other hand, by assuming perfect knowledge we make the two coincide, not only do we make the analysis unnecessarily unrealistic, but we introduce complications (by rendering the "imagined demand curve" indeterminate)
	
			(...)
	
			The "imagined" and the "real" demand curve are thus merely required to meet at one point - I admit, at the critical point.
		\end{quote}

		Here, the critical point is an outcome of the game: a pair of prices and a second-stage Nash equilibrium for those prices. If at least one group splits and no firm is at marginal cost, the implicit function theorem provides a coordination device for the first stage. A unique continuous imagined demand curve along which both price-setters can verify if deviations are profitable. If the standard first and second-order conditions are satisfied, this leads to a local SPE+. 

		What if, for example, we have at least one firm that always believes the response to price is discontinuous and favorable? That is, of course, a possibility when there are multiple second-stage equilibria, and, in that case, it's most likely there won't be any equilibria where both have profit. However, that seems to be a much stronger assumption requiring a particular economic motivation.

		We make no further assumption than the unique continuous deviation as a coordination device. Naturally, a necessary condition for a SPE+ outcome is negative marginal demand, but that is a characteristic of equilibria. Some authors assume negative marginal group demand, either explicitly (e.g., \citet*{Banerji:2009:LNE}) or implicitly through the distribution of preferences. However, this is an assumption, not a necessary condition, which in particular excludes some SPE+. Removing this assumption allows SPE+, for example, under positive network effects. For example, with two groups, the bandwagon effect can have two different directions (one for each group) where one compensates the other. The same dynamic may disrupt equilibria in the negative case. SPE+ is possible as long as the sum of all groups' marginal demand is negative. It may be counterintuitive why some groups' demand could increase with price (or why first stage players imagine so). We do not aim to discuss an economic motivation for why this may be so, but rather to point out that it is a possibility, which produces different equilibria. Probably the most notable discussion of a similar phenomenon is around \citet{Becker:1991:ANO}'s note on restaurant pricing. Becker's proposed explanation is studying particular forms for demand curves that would produce equilibria where one restaurant has excess demand and the other excess capacity. \citet*{Karni:1994:SAA} show that an SPE does not exist in such conditions when restaurants compete in prices. However, Becker's discussion is at the aggregate demand level. Our observation is that, on a group level, positive marginal demand is sustainable at equilibrium.

		The remainder of the work is as follows. In the next section we set up the model and present the main result. In Section \ref{s:dis} we discuss the results through models from the literature, and in Section \ref{s:conc} we conclude. We use $\equiv$ for definitions.

	\section{Model and main theorem.}
	For clarity we will use the nomenclature of firms and consumers, which form the set of players.
	In the first stage two firms $a$ and $b$, simultaneous and independently set a price for their product. In the second stage, consumers observe the vector $\mathbf{p}\equiv (p_a, p_b)$ and simultaneous and independently choose which firm to buy from.
	Consider a finite partition of the set of consumers into $g\in\mathbb{N}$ groups $\mathcal{G}\equiv\{G_1, \ldots, G_g\}$, and let $\mathbf{d}^j\equiv(d_1^j,\ldots, d_g^j)$ describe the number of consumers of each group that choose firm $j$. For each player in a group $G$ that chooses firm $j$, the payoff is\footnote{We are implicitly assuming a continuous set of consumers, and groups as a partition into measurable subsets. However, the results also hold in the finite case. We would just need to redefine payoffs as expected utility and work with mixed (behavior) strategies. The set of equilibria would be more complex, but the utility function would become multilinear. See, for example, our own work in \cite{Soeiro:2021:NNE}.}
	$$-p_j +v^j(\mathbf{d}^j; G), ~~j=a,b,$$
	where $v^j(\cdot)$ is a $C^2$ function capturing the influence groups may have on each other, i.e., network effects. Note it carries the assumption that within each group consumers are indistinguishable. However, the function may also encompass constant terms that represent how (all) elements of a group value each product (may include vertical or horizontal product differentiation), willingness to pay, or other.
	
	Suppose each group $i$ has mass (or size) $m_i$. We will denote by $\sigma_i\in [0,1]$ the fraction of group $i$ who choose firm $a$, hence, $d_i^a=m_i \sigma_i$ and $d_i^b=m_i(1- \sigma_i)$. We will whenceforth use this variable to represent the strategic action of each group, and let $\boldsymbol{\sigma}\equiv (\sigma_1, \ldots, \sigma_g)$ summarize consumption.\footnote{The assumption of mandatory consumption, or total coverage, is common in the literature and it simplifies the exposition. It is, however, not necessary, as we will discuss in the end.} Define now  $v: [0,1]^g\to \mathbb{R}^g$, a $C^2$ function where\footnote{It is a bit abusive to use the same letter, but allows us to economize notation and avoid carrying $\Delta$ around.}
	$$v_i(\boldsymbol{\sigma})\equiv v^a(\mathbf{m}\cdot \boldsymbol{\sigma}; G_i)-v^b(\mathbf{m}\cdot(\mathbf{1}-\boldsymbol{\sigma}); G_i)$$
	Given a pair of prices $\mathbf{p}$, whether $\boldsymbol{\sigma}$ is a second-stage Nash Equilibrium (NE) following $\mathbf{p}$, depends on $v$. It must hold:  $v_i(\boldsymbol{\sigma})\le p_a-p_b$ if $\sigma_i=0$; $v_i(\boldsymbol{\sigma})=p_a-p_b$ if $0<\sigma_i<1$; and $v_i(\boldsymbol{\sigma})\ge p_a-p_b$ if $\sigma_i=1$. As we will assume no cost function for firms (also a simplifying assumption, not necessary as briefly discussed in the final section), this is what determines the game above described, thus denoted by $\Gamma(v; g)\equiv\Gamma(v; g, \mathbf{m})$, and we will often omit $\mathbf{m}$, as group sizes aren't our main focus.
	
	\subsection{Equilibrium characterization (SPE)}
	In general, each game $\Gamma$ has a multiplicity of second stage Nash Equilibria (NE). Consider the correspondence $Q: (\mathbb{R}^+_0)^2 \to [0,1]^g$ where $Q(\mathbf{p})$ is the set of second stage NE following $\mathbf{p}$. Given $\mathbf{p^*}$ and a neighborhood $P(\mathbf{p^*})$, a \emph{local selection} of $Q$ is a function $q: P(\mathbf{p^*})\to [0,1]^g$ such that $q(\mathbf{p})\in Q(\mathbf{p})$ for all $\mathbf{p}\in P(\mathbf{p^*})$.
	For each outcome, to determine if deviations are profitable, each firm makes a local selection. Given $\mathbf{p^*}$, let $P_j$ be a neighborhood of $p_j^*$ for each $j=a,b$. Firm $a$ makes a local selection $q_a$ in neighborhood $P_a\times \{p_b^*\}$, and similarly firm $b$ makes a selection $q_b$ in $ \{p_a^*\}\times P_b$. Restricting each firm's selection to a neighborhood of its own price is not necessary but more standard to discuss the notion of SPE and unilateral deviation. Allowing firms to make different selections is more general than usual, but also not necessary. 
	
	Consider an outcome $(\mathbf{p^*}, \boldsymbol{\sigma})$. In a neighborhood $P$, (imagined) demand is
	$$D_a(p_a; q_a)=\mathbf{m}\cdot q_a(p_a, p_b^*), ~~~D_b(p_b; q_b)=\mathbf{m}\cdot (\mathbf{1}-q_b(p_a^*, p_b)).$$
	Note that, for a local selection, demand is a function of only one variable: own price. We assume firms have no costs, so profit for each firm along these selections is $\Pi_j(p_j; q_j)=p_j D_j$. (Note that these are all defined for a given outcome and some neighborhood of prices. At $\mathbf{p^*}$, we have actual outcome profits and demand, not imagined.)
	
	A triple $(\mathbf{p}^*, q_a, q_b)$ is a \textit{local} SPE in the standard way, that is, if for all $p_j$ in the neighborhood of $p_j^*$ it holds $(i)$ $q_a$ and $q_b$ are local selections (from $Q$); and $(ii)$ $\Pi_j(p_j; q_j)\le \Pi_j(p_j^*, q_j)$.
	
	\subsection{Results (SPE+)}
	Given a consumption profile $\boldsymbol{\sigma}$, we say that a group $i$ \textit{splits} if $0<\sigma_i<1$, denoting by $S(\boldsymbol{\sigma})\subseteq \mathcal{G}$ the groups that split, and $\bar{S}(\boldsymbol{\sigma})\subseteq \mathcal{G}$ those that don't. Thus, $\{S(\boldsymbol{\sigma}), \bar{S}(\boldsymbol{\sigma})\}$ is a partition of $\mathcal{G}$ induced by $\boldsymbol{\sigma}$, but for simplicity we will in general omit the dependence on $\boldsymbol{\sigma}$. We say that $\boldsymbol{\sigma}$ is a split if $S\neq\emptyset$, i.e., if at least one group splits. The next definition is based on the observation that for $(\mathbf{p}, \boldsymbol{\sigma})$ to be a SPE+ outcome it must be that $\boldsymbol{\sigma}\in Q(\mathbf{p})$ with at least one group that is indifferent ($v_i(\boldsymbol{\sigma})=\Delta p$). As we want to look at interior points we will focus on outcomes where $S(\boldsymbol{\sigma})\neq\emptyset$\footnote{This is in fact a necessary condition for equilibria where both firms have profit. Suppose $S=\emptyset$. If there are no indifferent consumers, then demand doesn't change in response to small price deviations. If there are indifferent consumers, but the groups to which they belong don't split (all choose the same firm), then one firm can provoke a jump for any small price deviation. Either way it would provide an incentive to deviations. This is general, not particular for this model. When the option not to buy is present, it must be that, if no group is indifferent between the two firms, there is a at least one group indifferent between buying or not. Nevertheless, we still need Definition \ref{d:ss} to force all indifferent groups out of $\bar{S}$.} but also where all groups that are indifferent between firms split.
	\begin{definition}[Stable splits]\label{d:ss}
		A split $\boldsymbol{\sigma}$ is \emph{stable} if, given any $i\in S$: $(i)$ for all $j\in S$, $v_{j}(\boldsymbol{\sigma})=v_i(\boldsymbol{\sigma})$; and $(ii)$ for all $j\in \bar{S}$, $v_{j}(\boldsymbol{\sigma})\neq v_i(\boldsymbol{\sigma})$.
	\end{definition}
	As discussed above, condition $(i)$ is necessary for second stage NE. Condition $(ii)$ is meant to focus on interior points, those where the NE inequality is strict.
	
	Consider now the restriction to $S$ of the Jacobian matrix of $v$, denoted $J_{ v}(\boldsymbol{\sigma};S)\equiv \left[(\partial v_i/\partial \sigma_{i'})(\boldsymbol{\sigma}), i,i' \in S\right]$. It can be seen as the marginal influence network among groups that split. For each group $i\in S$ define also the Hessian matrix restriction $H_{v_i}(\boldsymbol{\sigma};S)\equiv \left[\frac{\partial^2 v_i}{\partial \sigma_{i'}\partial \sigma_{j}}(\boldsymbol{\sigma}), i',j \in S\right]$. Let $l$ be the cardinal of $S$ (number of groups that split), and let us define for a given split $\boldsymbol{\sigma}$, the following two vectors, $\mathbf{k}\equiv(k_1,\ldots,k_l)$ and $\mathbf{r}\equiv(r_1,\ldots,r_l)$ where
	$$k_i(\boldsymbol{\sigma};S)\equiv\dfrac{\sum_{i'}C_{i'i}(J_{ v})}{\det(J_{ v})}, ~~~ r_i(\boldsymbol{\sigma};S)\equiv\dfrac{\sum_{i'}\mathbf{k} H_{v_{i'}}\mathbf{k}^{T} C_{i'i}(J_{ v})}{\det(J_{ v})}.$$
	and where $C_{ii'}(\cdot)$ denotes the $ii'-$cofactor (for simplicity we have omitted the dependence on $(\boldsymbol{\sigma};S)$ for all vectors and matrices in the expression). The vectors $\mathbf{k}$ and $\mathbf{r}$ capture the first and second order reaction to changes in consumption among groups that split. Let finally
	$$K_S(\boldsymbol{\sigma})\equiv \sum_{i\in S} m_i k_i(\boldsymbol{\sigma};S), ~\textrm{and} ~ R_S(\boldsymbol{\sigma})\equiv \sum_{i\in S} m_i r_i(\boldsymbol{\sigma};S).$$
	
	\begin{lemma}\label{l:uq}
		For an outcome $(\mathbf{p}^*, \boldsymbol{\sigma})$ where $\boldsymbol{\sigma}$ is a stable split in $Q(\mathbf{p^*})$, if $J_v(\boldsymbol{\sigma};S)\neq 0$, there is a unique continuous local selection $q$ such that $q(\mathbf{p^*})=\boldsymbol{\sigma}$.
		Furthermore, $q$ is $C^2$, $D'_j(p_j^*; q)=K_S(\boldsymbol{\sigma})$ and $D''_j(p_j^*; q)=R_S(\boldsymbol{\sigma})$.
	\end{lemma}
	The response of demand to prices is given by the interdependence of groups that split. Stable splits guarantee that $\bar{S}$ is less sensitive, so that when the choice of a group is cohesive, it remains unchanged for small price deviations. The proof uses an implicit function theorem, and is left for the appendix. Let us omit the dependence on $\boldsymbol{\sigma}$ by writing $K_S$ and $R_S$.
	\begin{definition}\label{d:real}
		A split $\boldsymbol{\sigma}$ is \emph{realizable} if
		\begin{enumerate}
			\item $K_S<0$;
			\item $-\frac{1}{\mathbf{m}\cdot (\mathbf{1}-\boldsymbol{\sigma})}<\dfrac{R_S}{2K_S^2}<\frac{1}{\mathbf{m}\cdot \boldsymbol{\sigma}}.$
		\end{enumerate}
	\end{definition}
	A split is realizable if the first and second-order conditions ensure the split leads to a local maximum for the first stage along the unique continuous local selection. Note that for realizable splits \textit{not to exist} in a given game, a condition much stronger than may appear at first must hold. It would be necessary that, for every $S^*\subseteq\mathcal{G}$ and every $\boldsymbol{\sigma}$ such that $S(\boldsymbol{\sigma})=S^*$, condition $1$ or $2$ is not satisfied.  In particular, as $K_S$ is the sum of all $k_i$ for $i\in S$, consistently breaking condition $1$ imposes a restriction on all possible combinations. For example, in the linear case, where $R_S=0$ and $K_S$ is constant (for each $S$), it is sufficient that at least for one group $G$ we have $K_{G}<0$. In that case, a profile where only that group splits is a realizable split. We will discuss this further in the next section.
	
	Next, we state the main result. Let us first define the following price function,
	$$\psi(\boldsymbol{\sigma})=\frac{1}{-K_S}\left(\mathbf{m}\cdot \boldsymbol{\sigma},  \mathbf{m}\cdot (\mathbf{1}-\boldsymbol{\sigma})\right).$$

	\begin{theorem}\label{t.main}
		Let $\boldsymbol{\sigma}$ be a realizable stable split.\\
		If $\boldsymbol{\sigma}\in Q(\psi(\boldsymbol{\sigma}))$ then $(\mathbf{p}^*, q(p_a, p_b^*), q(p_a^*, p_b))$ is a local SPE with profit for both firms, where $\mathbf{p^*}=\psi(\boldsymbol{\sigma})$ and $q$ is the continuous local selection associated with the outcome $(\mathbf{p}^*, \boldsymbol{\sigma})$.
	\end{theorem}
	\begin{proof}
		Suppose $\boldsymbol{\sigma}$ is a realizable stable split in $Q(\psi(\boldsymbol{\sigma}))$. Then, Lemma \ref{l:uq} is applicable. This guarantees that (standard profit derivation along $q$)  $\mathbf{p}^*$ is a critical point. As it is realizable, conditions $1$ and $2$ of Definition \ref{d:real} guarantee it is a local maximum for both firms. As it is a second-stage Nash, it is a local SPE+.
	\end{proof}
	
	\section{Discussion and examples}\label{s:dis}
	Two natural questions arise from Theorem \ref{t.main} (essentially amounting to whether it is empty):
	$(i)$ does a given game $\Gamma$ have realizable splits?
	$(ii)$ is at least one of these splits stable and a solution to $\boldsymbol{\sigma}\in Q(\psi(\boldsymbol{\sigma}))$?
	
	The answer to $(i)$ depends on the marginal properties of $v$ and is the main focus of this work. Question $(ii)$ amounts to whether there is a realizable stable split solution to
	\begin{equation}\label{e.s0}
		v_i(\boldsymbol{\sigma})=\mathbf{m}\cdot(2\boldsymbol{\sigma}-\mathbf{1})/K_S(\boldsymbol{\sigma}), ~ \forall i\in S.
	\end{equation}
	The right-hand side of Equation \ref{e.s0} is the difference of equilibrium prices $p_a^*-p_b^*$. Satisfying the equation means $\boldsymbol{\sigma}$ is a second-stage NE for $\mathbf{p^*}$. Question $(i)$ ensures first-stage (pricing game) solutions, while $(ii)$ second-stage (consumption game) solutions compatible with $(i)$. This section of the work is divided into two subsections where we discuss the above two questions through models in the literature.
	
	\subsection{(i) Realizable splits.}
	\subsubsection{One group, $g=1$.}
	Let us start by looking at games where there is only one group $\Gamma(v;1)$. All NE splits are stable. Now, we have $K(\sigma)=m(v'(\sigma))^{-1}$ and if $v''(\sigma)\neq 0$, then $R(\sigma)=\frac{K^2 v'(\sigma)}{2v''(\sigma)}$; if $v''(\sigma)=0$, then $R(\sigma)=0$. As such when $v$ is non-decreasing (positive network effects) there are no realizable splits. If $v$ is decreasing for at least some point (negative network effects), realizable splits may exist. In particular, if $v$ is linear, $K$ and $R$ are constant and $R=0$. If $K<0$ all splits are realizable.
	
	\begin{example}[$\Gamma(v;1)$; \cite{Grilo:2001:PCW}]
		The authors consider the following case: $v^j(d^j)=\alpha d^j-\beta (d^j)^2-t(x-x_j)^2$ with mass of consumers $m$. Let us consider the case with no product differentiation $t=0$. Then, $v(\sigma)=(2\sigma-1)(\alpha m-\beta m^2)$, $K=[2(\alpha-\beta m)]^{-1}$ and $R=0$. As such, whenever $\alpha<\beta m$ every split is realizable, that is, when negative effects overcome positive ones. In that case, there is a unique realizable split solution to Equation \ref{e.s0}, $\sigma=1/2$. If there are only positive network effects, $\beta=0$, there are no realizable splits.
	\end{example}

	Some variations with only one group include a dependence of $\alpha$ on the firms. \citet{Tolotti:2020:HBD}, in the case of $t=0$, discuss $v^j(d^j)=\alpha^j d^j$. In this case $v(\sigma)=(\alpha^a+\alpha^b) m\sigma-\alpha^b m$. Hence, $K=(\alpha^a+\alpha^b)^{-1}$. Again, no realizable split unless negative effects overcome positive effects. Assuming $\alpha^a+\alpha^b< 0$, all splits are realizable and stable. However, there is a unique solution to Equation \ref{e.s0}, $\sigma=\frac{\alpha^a}{\alpha^a+\alpha^b}$ leading to equilibrium prices $p_j=-m\alpha^j$, so it is necessary that $\alpha^j<0$, for $j=a,b$.
	
	The general idea is that, with $g=1$, if positive network effects dominate, ceteris paribus, there are no realizable splits. 
	\subsubsection{Multilinear case, $g>1$.}
	Consider the multilinear extension for more groups,
	$$v_i^j(\mathbf{d}^j)=\sum_{i'=1}^{g}\alpha_{ii'}^j d_{i'}^j$$
	where $\alpha_{ii'}^j\in\mathbb{R}$ represents how a member of group $i$ is influenced by a member of group $i'$ when both choose firm $j$.
	Let us define $w_{ii'}\equiv \alpha_{ii'}^a +\alpha_{ii'}^b$. Note that
	$$v_i(\boldsymbol{\sigma})=\sum_{i'=1}^{g}m_{i'}w_{ii'}\sigma_{i'}-\sum_{i'=1}^{g}\alpha_{ii'}^b m_{i'}.$$
	Thus, $\frac{\partial v_{i}}{\partial \sigma_{i'}}= m_{i'}w_{ii'}$. One of the advantages of the linear case is that $K_S(\boldsymbol{\sigma})$ and $R_S(\boldsymbol{\sigma})$ do not depend directly on $\boldsymbol{\sigma}$ but only on $S$, the groups that split. Furthermore, $R_S=0$ for all $S$. We separate $m_i$ from $\alpha$ because $K_S$ does not depend on $m_i$ either\footnote{This is because the product of $m_i$'s cancel out in the numerator and denominator expression for $K$. However, each $k_i$ depends on $\mathbf{m}$. This will be clear in Example \ref{ex:g2}.}, and it is useful to work with the matrix $W\equiv W(\mathcal{G})\equiv(w_{ii'})$.  As any $K_S$ is determined by a submatrix of $W$, a multilinear game $\Gamma$ is determined by $W$ and we can identify the specific game considered by $\Gamma(W; g)$.
	
	\begin{example}[$\Gamma(W; 2)$]\label{ex:g2}
		An advantage of $g=2$ is that there are only two types of splits. \emph{Total}, where both groups split, and \emph{singular}, where only one group splits. Let $\mathcal{G}=\{G_1, G_2\}$, $\mathbf{m}=(m_1,m_2)$ and
		$$
		W\equiv 
		\begin{pmatrix}
			w_{11}& w_{12}\\ 
			w_{21}	 & w_{22}
		\end{pmatrix}.
		$$
		For total splits,
		$$K_\mathcal{G}=\frac{(w_{11}+w_{22})-(w_{12}+w_{21})}{w_{11}w_{22}-w_{12}w_{21}}.$$
		For singular splits, $K_{G_i}=(w_{ii})^{-1}$. As $R_S=0$, realizable splits require only that $K_S<0$. However, recall that $K_S= \sum_{i\in S} m_i k_i(S)$. Here,
		$$k_1=\frac{(w_{22}-w_{12})m_2}{(w_{11}w_{22}-w_{12}w_{21})m_1m_2}, ~~k_2=\frac{(w_{11}-w_{21})m_1}{(w_{11}w_{22}-w_{12}w_{21})m_1m_2},$$
		which need not necessarily be both negative. Moreover, it shows why $K_\mathcal{S}$ is independent of $\mathbf{m}$
		
		Although realizable singular splits require negative network effects, i.e., negative entries in $W$, it is possible to have realizable (total) splits even when there are only positive entries. For example, consider
		$$
		\underline{W}=
		\begin{pmatrix}
			1& 2\\ 
			3	 & 5
		\end{pmatrix},
		$$
		Note that $k_1=-3/m_1$, $k_2=2/m_2$ and $K_{\mathcal{G}}=-1$. Hence total splits are realizable (the unique type of realizable splits), and this is independent of $\mathbf{m}$. If we change the sign in all entries on $\underline{W}$, then the results reverse. Total splits would not be realizable, but singular splits would.
	\end{example}
	
	\begin{example}[\cite{Amaldoss:2005:POC}]
		The following case is considered (here we have assumed no product differentiation): $\mathcal{G}=\{l,c\}$, group sizes are $m_l=\beta$ and $m_c=1-\beta$ for $\beta\in (0,1)$, and
		$$v_l^j=-\lambda_l d^j, ~~~ v_c^j=\lambda_c d^j,$$
		where $\lambda_G>0$. Group $l$ are called snobs/ leaders and group $c$ conformists/ followers. Observe that although there are two groups, both look at the same aggregate. Let's look at the following modified game $\Gamma_{\delta}(W,2)$ where
		$$
		W=
		\begin{pmatrix}
			-\lambda_l		& -\lambda_l\\ 
			\lambda_c+\delta	   & 	\lambda_c
		\end{pmatrix}.
		$$
		Network effects in \cite{Amaldoss:2005:POC} correspond to $\delta=0$. The authors use a bound on conformism coupled with an assumption on the distribution of preferences to ensure a SPE+ (we discuss this further in the next subsection). The case of no product differentiation corresponds to a game $\Gamma_0(W,2)$. In this case, total splits are not realizable, because the $\det(W)=0$ so $K_{\mathcal{G}}$ is not defined. However, $K_l=(-\lambda_l)^{-1}$ and singular splits where only snobs ($l$) split are realizable. This is similar to a game with only snobs.
		
		In the modified game we propose, i.e., when $\delta\neq0$, conformists distinguish between snobs and other conformists. So there are two aggregate variables in play as $d^j$ becomes $\mathbf{d}^j=(d_l^j, d_c^j)$. If $\delta>0$ conformists prefer snobs (leaders), if $\delta<0$, conformists prefer other conformists (to maintain the game idea we assume $\delta>-\lambda_c$). Note that, for $\delta\neq0$, we have $K_{\mathcal{G}}=(-\lambda_l)^{-1}$. So, total splits are now also realizable. Interestingly, realizability will not depend on $\delta$ nor on $\lambda_c$, the influence among conformists. However, although $K_{\mathcal{G}}=K_l$, outcomes will differ for the two types of split. For singular splits, if all conformists choose $a$, $p_a^*=\lambda_l(1-\beta\sigma_l)$; and if all conformists choose $b$, $p_a^*=\lambda_l\beta\sigma_l$. For total splits, $p_a^*=\lambda_l(\beta\sigma_l+(1-\beta)\sigma_c)$. In both cases, prices are proportional to the degree of snobbishness.
	\end{example}

	\begin{example}[\cite{Armstrong:2006:CIT}]
		The following structure is assumed,
		$$
		W=
		\begin{pmatrix}
			0& w_{12}\\ 
			w_{21}	 & 0
		\end{pmatrix}, ~~\textrm{thus}~~K_\mathcal{G}=\frac{w_{12}+w_{21}}{w_{12}w_{21}},
		$$
		where $w_{12}, w_{21}>0$. There are no realizable splits. The  Hotelling specification is used, with costs $t_1$ and $t_2$ for each group, which coupled with the assumption that $16 t_1t_2> (w_{12}+w_{21})^2$ produces SPE+. That is when network (externality) parameters are small compared to differentiation ones.\footnote{Note that this is consistent with Armstrong's model. The author considers $m_1=m_2=1$. We get $p_a^*=(\sigma_1+\sigma_2) \frac{-w_{21}w_{12}}{w_{21}+w_{12}}<0$, thus no equilibria, and exactly the price expression presented there if one considers no transportation costs (observe that $w_{ij}=2\alpha_j$ for the parameters $\alpha_j$ considered there).}
		
		Suppose we change the structure, for example, by allowing intra-group $2$ influence in the following way,
		$$
		W=
		\begin{pmatrix}
			0& w_{12}\\ 
			w_{21}	 & w_{12}+w_{21} +\delta
		\end{pmatrix}, ~~\textrm{then}~~K_\mathcal{G}=\frac{\delta}{-w_{12}w_{21}},
		$$
		for some $\delta>0$. Total splits are now realizable. Note that it is necessary that the intra-group influence introduced overcomes the sum of cross-group.
		
		However, if one wishes a model with a zero diagonal matrix (no intra-group influences), then at least 3 groups are necessary to guarantee the existence of realizable splits. Suppose you have 3 groups and consider the following structure,
		$$
		W=
		\begin{pmatrix}
			0& 2r & r\\ 
			z& 0&\frac{rz}{4r+z}-\varepsilon \\ 
			\frac{z}{2} &  0 & 0 \\ 
		\end{pmatrix},
		$$
		where $r,z>0$ and $0<\varepsilon<\frac{rz}{4r+z}$.  The entries, and thus network effects, are all positive. However, $K_{\mathcal{G}}<0$ and the total split is realizable.
	\end{example}

	\begin{example}[Adjacency matrix.]\label{ex:adj}
	In some models, a graph is given a priori representing the connections among groups. For example \cite{Banerji:2009:LNE} consider (for the case of incompatible products),
	$$v_i^j(\mathbf{d}^j)=\sum_{i'\in N(i)}d_{i'}^j,$$
	where $N(i)$ are the neighbors of $i$, including itself. This is an instance of the multilinear case, where $\alpha_{ii'}^a=\alpha_{ii'}^b\equiv\alpha_{ii'}$; $\alpha_{ii'}=\alpha_{i'i}$; and $\alpha_{ii'}\in\{0,1\}$. The latter reflects whether groups $i$ and $i'$ are connected. Let $A$ be the adjacency matrix of the given (unweighted and undirected) graph. As $w_{ii'}=2\alpha_{ii'}$, $W=2A$, thus $\det(W)=2^{g}\det(A)$, and $C_{ii'}(W)=2^{(g-1)}C_{ii'}(A)$. Hence, given $K_S^*$ in a game $\Gamma(W,g)$, for the game $\Gamma(A,g)$ we have $2K_S=K_S^*$. As such, to check if realizable splits exist, it suffices to study the adjacency matrix.\footnote{Note that the same holds for $\mathbf{k}=2\mathbf{k}^*$}
	
	Naturally, whether realizable splits exist depends on the type of network under consideration. In particular, for the total split we need to look at the whole adjacency matrix, for partial splits, it is necessary to look at the adjacency matrix of the subgraph induced by $S\subseteq\mathcal{G}$. For example, the complete network has $\det(W_{complete})=0$, and as any induced subgraph is also complete, $K_S$ is not defined for any $S$ and no realizable split exists, whether partial or total. For the star network with loops, $K_{\mathcal{G}-star}=\frac{1}{2}$ (for $\Gamma(W_{star};g)$), thus total splits are not realizable.\footnote{This can be shown by induction, but we omit the proof here.} A subgraph is either again a star network if the central node splits (is in $S$), or a graph composed by loops only. In the latter, the relevant matrix has only 1's in the diagonal, leading to $K_S=\frac{1}{2}$. So for the star network there are no realizable splits. In Figure \ref{fig:netex} is an example of a network where total splits are realizable. We have made a small modification from a star network.
	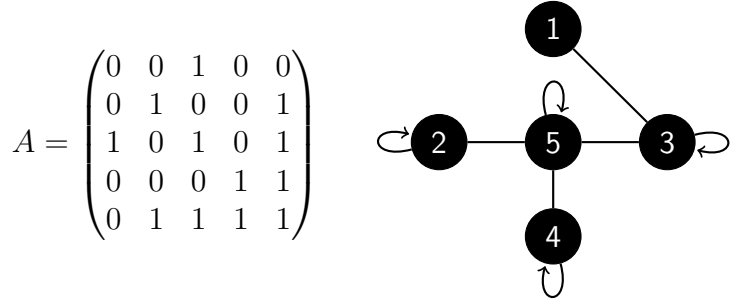
\begin{figure}[htb]
		\centering
		$\displaystyle
		A=
		\begin{pmatrix}
			0 & 0 & 1 & 0 & 0 \\
			0 & 1 & 0 & 0 & 1 \\
			1 & 0 & 1 & 0 & 1 \\
			0 & 0 & 0 & 1 & 1 \\
			0 & 1 & 1 & 1 & 1
		\end{pmatrix}
		~~~~
		\begin{tikzpicture}
			[baseline=(2.base),bn/.style={circle,fill,draw,text=white,font=\sffamily, minimum size=6mm}, every node/.append style={bn}]
			\node (1) at (2.5,3) {1};  
			\node (2) at (1,1.5)  {2};  
			\node (3) at (4,1.5)  {3};  
			\node (4) at (2.5,0.25) {4};  
			\node (5) at (2.5,1.5)  {5};
			\draw[thick] (4)--(5) (3)--(5) (2)--(5) (1)--(3);
			\path
				(2) edge [loop left, thick]  (2)
				(3) edge [loop right,, thick]  (3)
				(4) edge [loop below, thick]  (4)
				(5) edge [loop above, thick]  (5);
		\end{tikzpicture}
	$
		\caption{A network inducing realizable total splits for $\Gamma(2A;5)$. Considering $W=2A$ we get $K_{\mathcal{G}}=-0.5$ and $\mathbf{k}=(-1/m_1, -0.5/m_2, 0.5/m_3, -0.5/m_4, 1/m_5)$.}
		\label{fig:netex}
	\end{figure}
	
		The example in Figure \ref{fig:netex} is not unique, nor specific to $g=5$. We did however need to increase the number of groups to $5$ to find a solution. If $g>5$ solutions exist too and become easier to find. Moreover, $A$ could be the base for a realizable partial split in a game with $g>5$. We have chosen a symmetric matrix to fall within the assumption of an undirected (and unweighted) network. However, as we will see, being unweighted forces an increase in dimension to find an example, while being undirected produces symmetric outcomes.
	
		The fact that realizable splits exist in such cases might appear to contradict the result in \cite{Banerji:2009:LNE}. Their Lemma 1 states no split can be part of an SPE+. However, the result comes from their assumption A1 that implies $k_i\le 0$ for all $i$ and all $S$. The assumption that $K_S <0$ can only be obtained this way, eliminates all realizable splits in their class of games. While it may be a reasonable economic assumption, from a mathematical point of view it is crucial to be aware that it is in fact a strong imposition, which eliminates a rich set of equilibria where both firms have profit. In particular, it is an imposition on how firms imagine the demand response. The implicit function theorem provides unique local selections for stable splits, where it may happen $k_i>0$ for some $i\in S$, but as long as $K_S<0$ and the split is realizable, an equilibrium is possible. That is, equilibrium selections exist where it does not hold $k_i\le0$ for all $i$.
	\end{example}

	\subsection{(ii) NE solutions.}	
		The second question arising from Theorem \ref{t.main} is whether a realizable split will produce a SPE outcome in a game $\Gamma(v;g)$. The answer is: if it is stable and a solution to Equation \ref{e.s0}. Whether for a given game the system induced by Equation \ref{e.s0} has such a solution is a more standard existence question, but out of scope here. Our objective is to discuss how the structure of network effects among groups affect SPE, even if members of the same group are indistinguishable. In this section our goal is to show that structures producing realizable splits, also produce SPE for some classes of games.
		
		Given $\boldsymbol{\tau}\in \mathbb{R}^g$ and $\varepsilon >0$, let us define
			$$\underline{v}_i(\boldsymbol{\sigma})\equiv\underline{v}_i(\boldsymbol{\sigma}; [\boldsymbol{\tau}, \varepsilon])\equiv
		\left\lbrace 
		\begin{array}{lr}
			v_i(\boldsymbol{\sigma})-\tau_i +\varepsilon &  \text{if} ~~~ \sigma_i=1 \\
			v_i(\boldsymbol{\sigma})-\tau_i & \text{if} ~~~ i\in S(\boldsymbol{\sigma}) \\
			v_i(\boldsymbol{\sigma})-\tau_i -\varepsilon &  \text{if} ~~~ \sigma_i=0
		\end{array}
		\right.
		.
		$$
		A split is realizable in $\Gamma(v;g)$ if, and only if, it is realizable in $\Gamma(\underline{v};g)$. 
		\begin{remark}
			Let $\boldsymbol{\sigma}$ be a realizable split of a game $\Gamma(v;g)$ and $\mathbf{p^*}=\psi(\boldsymbol{\sigma})$. There is a unique $\boldsymbol{\tau}$  such that for every $\varepsilon>0$,  $(\mathbf{p^*},\boldsymbol{\sigma})$  is a SPE outcome of the game $\Gamma(\underline{v};g)$. Namely,
			$$\tau_i = v_i(\boldsymbol{\sigma})-\mathbf{m}\cdot(2\boldsymbol{\sigma}-\mathbf{1})/K_S(\boldsymbol{\sigma}).$$
		\end{remark}
		It is a remark because we are simply modifying the game by readjusting the system, so that its solution is in $[0,1]^g$ and interior (stable split). Network effects are the same independently  of $[\boldsymbol{\tau}, \varepsilon]$, which is the focus of this work. Note that, for example for $i\in S$, $\tau_i$ can be seen as coming from $\tau_i^a-\tau_i^b$ for each group, representing product differentiation, willingness to pay, or other constant added to the utility of consumers (similarly for $i\in \bar{S}$ if $\varepsilon$ is included). In fact, the remark defines a class of games $\Gamma(\underline{v}; \boldsymbol{\tau})$ for every $\varepsilon>0$. If we write $\varepsilon$ dependent on $i$, the same remark holds, and it would be natural to see it as a distribution of preferences that induces $\underline{v}$ where some groups are the indifferent consumers. The more standard approach, however, is to consider a distribution for each group and $v_i$ as the value of consumption for the indifferent consumer of that group. Here we are always assuming individuals of the same group are indistinguishable a priori. Let us just emphasize what the remark states: every realizable split is a SPE+ outcome for a class of games $\Gamma(\underline{v}; \boldsymbol{\tau})$ with network effects characterized by $v$.
		
		Which split is a SPE for the class $\boldsymbol{\tau}=\mathbf{0}$, if it exists, depends on $v$. In the class of multilinear games, discussed before, all total splits are realizable when $K_{\mathcal{G}}<0$. However, the solution to Eq.\ref{e.s0} is unique (it might be a real solution for the game or not), so only one realizable split will lead to an SPE. The fact that it exists and is unique comes from the multilinearity of the determinant operator (and because it is not zero for realizable splits). The system of $g$ equations induced by Eq.\ref{e.s0} can be rewritten as the following equation for each group $i$,
		\begin{equation}\label{eq:sys}
			\sum_j \left( w_{ij}-\frac{2}{K_{\mathcal{G}}}\right)m_j\sigma_j = \sum_j \left(\alpha_{ij}^b-\frac{1}{K_{\mathcal{G}}} \right)m_j.
		\end{equation}
		A gross (but educated) guess (and hence, the solution if it works) is
		\begin{equation}\label{eq:sol}
			\sigma_j=\frac{1}{2}\frac{\alpha_{ij}^b-\frac{1}{K_{\mathcal{G}}}}{\frac{w_{ij}}{2}-\frac{1}{K_{\mathcal{G}}}}.
		\end{equation}
		It generally doesn't work. However, recalling that $w_{ij}\equiv\alpha_{ij}^a+\alpha_{ij}^b$, it shows the importance of $\alpha_{ij}^a\neq\alpha_{ij}^b$ in producing price dispersion and asymmetric demand (when no other differentiation is present). We can write the following proposition.
		\begin{proposition}
			In a multilinear game, given a group $G_j$, if for all $G_i$ it holds $\alpha_{ij}^a=\alpha_{ij}^b$, then, in a total split SPE+ outcome, $\sigma_j=0.5$.
		\end{proposition}
	\begin{proof}
		If Equation \ref{eq:sol} is independent of $i$ then it is the solution to the system generated by Equation \ref{eq:sys}. If $\alpha_{ij}^a=\alpha_{ij}^b$ then $\sigma_j=0.5$.
	\end{proof}
		The proposition can be rephrased in the following way: If a member of group $G$ has the same influence on all other players with whom it is connected (including those of its own group), then, in equilibrium, $G$ must split in half. The observation emphasizes the importance of (asymmetric) weights relative to the underlying network structure and how imposing strong symmetry properties leads to symmetric outcomes.
		\begin{example}	
			In the class of games based on an adjacency matrix, discussed in Example \ref{ex:adj}, $\alpha_{ij}^a=\alpha_{ij}^b \in\{0,1\}$. It is also assumed the network is undirected, i.e., we get $w_{ij}=2\alpha_{ij}=2\alpha_{ji}$. As such, the guess leads to the solution, $\boldsymbol{\sigma}=\mathbf{0.5}$ (which is independent of $\mathbf{m}$, $g$ and the network as long as $K_{\mathcal{G}}<0$). Let $M=\sum_i m_i$. Equilibrium prices are $p_a^*/M=p_b^*/M=1$.
			
			Naturally, if one wishes an outcome where prices and demand differ while keeping the network undirected, it is possible to choose $\boldsymbol{\tau}$ so that the modified game $\Gamma(\underline{v};\boldsymbol{\tau})$ has the desired equilibrium.
		\end{example}

	\section{Conclusion}\label{s:conc}
	How demand responds to prices determines the outcome of price competition. However, demand is an aggregate of consumption decisions. When different groups interact and influence each other, their response need not be homogeneous and mimic the aggregate behavior.  It is thus not necessary that consumption jumps from corner solution to corner solution when network effects induce bandwagon responses to price changes. Nevertheless, a plethora of possibilities for the composition of the aggregate need not bring uncertainty. Network effects produce an interlocking of decisions that stabilizes price competition. These are the main points of this work.
	
	We have considered the partition into groups given a priori, but it can be endogeneized. One way is as a classification tool, based on the characteristics of consumers that produce the same set of equilibria (see, for example, \citet{Soeiro:2015:EEI}). Another way is to consider the partition an assessment by the firms, which could have an associated cost investment decision. It would be a possible extension of the work: what happens when firms use different partitions to decide prices?
	
	If it is a choice of information structure, it seems natural that the appropriate partition into groups lies somewhere between no information (just aggregate behavior) or complete information (network effects at the individual level). Naturally, the latter may be more accurate, but either unfeasible or with a high cost for firms. If a firm invests in studying individual-level network effects, it is natural it wants to price discriminate to build value for that information. So the results obtained seem appropriate for uniform pricing: at least two groups when the asymmetry of weights and direction is allowed, or at least five groups for undirected and unweighted structures.
	
	Two simplifying assumptions may concern the reader: mandatory consumption and no costs for firms. The essential mathematical tool used is the implicit function theorem. The result is local and, the function behind it is the difference between buying from one or the other firm. With a third option, this can be extended, for each group, to three sets by including the following two splits: difference of buying or not from each firm. Removing mandatory consumption may change some outcomes, but it will fit the implicit function theorem without a particular problem. Introducing a cost function for firms is even less problematic, as we have based results on second-stage interactions. A fixed cost may eliminate some equilibria (if profits don't cover them), and a linear cost function provokes a translation in prices, but results follow straightforwardly. A more general cost function will have the usual difficulties and need their counterpart assumptions.
\appendix
	\section{Proof of Lemma \ref{l:uq}}

		\textit{For an outcome $(\mathbf{p}^*, \boldsymbol{\sigma})$ where $\boldsymbol{\sigma}$ is a stable split in $Q(\mathbf{p^*})$, if $J_v(\boldsymbol{\sigma};S)\neq 0$, there is a unique continuous local selection $q$ such that $q(\mathbf{p^*})=\boldsymbol{\sigma}$.
		Furthermore, $q$ is $C^2$, $D'_j(p_j^*; q)=K_S(\boldsymbol{\sigma})$ and $D''_j(p_j^*; q)=R_S(\boldsymbol{\sigma})$}.
	\begin{proof}
		Let $(\mathbf{p}^*, \boldsymbol{\sigma}^*)$ be an outcome where $\boldsymbol{\sigma}^*$ is a stable split in $Q(\mathbf{p^*})$. Recall that it holds
		\begin{enumerate}[(i)]
			\item if $\sigma_i=1$ then $v_i(\boldsymbol{\sigma}^*)>p_a^*-p_b^*$;
			\item if $\sigma_i=0$ then $v_i(\boldsymbol{\sigma}^*)<p_a^*-p_b^*$;
			\item if $0<\sigma_i<1$ then $v_i(\boldsymbol{\sigma}^*)=p_a^*-p_b^*$.
		\end{enumerate}
		Hence, for groups that satisfy $(i)$ or $(ii)$, there is some neighborhood $N(\mathbf{p}^*, \boldsymbol{\sigma}^*)$ where the best response of their members is constant. Recall that $l\equiv\#S$ denotes the number of groups that split, hence, the number of groups that satisfy $(iii)$. Note that $l\ge1$ (stable split, see Definition \ref{d:ss}). Index the groups that split by $S(\boldsymbol{\sigma}^*)\equiv \{1,\ldots, l\}$. Denote by $\mathbf{s}^*\equiv(s_1^*,\ldots,s_l^*)$ their respective coordinates in $\boldsymbol{\sigma}^*$ and let $\boldsymbol{\sigma}^*_{\bar{S}}$ denote the profile of the remaining groups. Consider now the function induced by the utility of each consumer in any group in $S$ and given by
		$$\Delta u_i(\mathbf{p},\mathbf{s})\equiv\Delta u_i(\mathbf{p},\mathbf{s};\boldsymbol{\sigma}^*_{\bar{S}})\equiv v_i(\boldsymbol{\sigma})-\Delta p$$
		where $\Delta p \equiv p_a-p_b$, and define
		\begin{eqnarray*}
		\Delta U :(\mathbb{R}_0^+)^2\times (0,1)^{l} &\to &\mathbb{R}^{l}\\
		(\mathbf{p},\mathbf{s}) & \mapsto & \Delta u_1(\mathbf{p},\mathbf{s}), \ldots, \Delta u_l(\mathbf{p},\mathbf{s}).
		\end{eqnarray*}
		Observe that $\Delta U(\mathbf{p}^*,\mathbf{s}^*)=\mathbf{0}$. By $(iii)$ and because the best response for members of groups in $\bar{S}$ is locally constant, any outcome $(\mathbf{p},\mathbf{s}, \boldsymbol{\sigma}^*_{\bar{S}})$ in the neighborhood $N(\mathbf{p}^*, \boldsymbol{\sigma}^*)$ such that $\Delta U(\mathbf{p},\mathbf{s})=\mathbf{0}$, is a Nash equilibrium for consumers given $\mathbf{p}$.
		Note now that $\Delta U$ is $C^2$ and $J_{\Delta U, \mathbf{s}}(\mathbf{p}^*,\mathbf{s}^*; \boldsymbol{\sigma}^*_{\bar{S}})=J_{v}(\boldsymbol{\sigma}^*; S)$. Therefore, if $\det\left[ J_{v}(\boldsymbol{\sigma}^*; S)\right]\neq 0$, using the Implicit Function Theorem, there is a product neighborhood $N_p\times N_{s} \subset (\mathbb{R}_0^+)^2\times (0,1)^{l}$, containing $(\mathbf{p}^*,\mathbf{s}^*)$, and a unique $C^2$ function $\phi: N_p \to N_{s} $ such that $\Delta U(\mathbf{p}, \phi(\mathbf{p}))=\mathbf{0}$. Furthermore, in that neighborhood
		$$\sum_{i=1}^{l}\dfrac{\partial \Delta U}{\partial \phi_i(\mathbf{p})}(\mathbf{p}, \phi(\mathbf{p}))\dfrac{\partial \phi_i}{\partial p_a}(\mathbf{p})=\dfrac{\partial \Delta U}{\partial p_a}(\mathbf{p},\phi(\mathbf{p})).$$
		As in the same neighborhood $\Delta u_i(\mathbf{p}, \boldsymbol{\sigma})=0$, we get $v_i(\boldsymbol{\sigma})=\Delta p$, and
		$$ J_{v}(\phi(\mathbf{p}); S)\dfrac{\partial \phi}{\partial p_a}(\mathbf{p})=\mathbf{1}_{l \times 1}.$$
		Using the Cramer rule,
		$$\dfrac{\partial \phi_i}{\partial p_a}(\mathbf{p})=\dfrac{\det\left[ J^{(i)}_{v}(\phi(\mathbf{p}); S)\right]}{\det\left[ J_{v}(\phi(\mathbf{p}); S)\right]},$$
		where $J^{(i)}_{\Delta v}(\phi(\mathbf{p}); S)$ is obtained by replacing the entries of column $i$ with $1$ in $J_{v}(\boldsymbol{\sigma}^*; S)$.
	
		Therefore, for the outcome $(\mathbf{p}^*, \boldsymbol{\sigma})$, given a neighborhood $N$ and a continuous local selection $q$, demand is
		$$D_a(p_a; p_b, q)=\sum_{i\in\mathcal{G}}m_i q_i(\mathbf{p})=\sum_{i\in S} m_i \phi_i(\mathbf{p}) + \underline{D_a}(\boldsymbol{\sigma}^*_{\bar{S}}),$$
		where $\underline{D_a}(\boldsymbol{\sigma}^*_{\bar{S}})$ is locally a constant: the number of groups that choose $a$ (i.e., satisfy $(i)$). As such in that neighborhood (because groups that don't split do not change)
		\begin{equation}\label{e:p:dd}
			D'(p_a; p_b, q)=\dfrac{\partial D_a}{\partial p_a}(\phi(\mathbf{p}))=\sum_{i}m_i\dfrac{ \det\left[ J^{(i)}_{v}(\phi(\mathbf{p}); S)\right]}{\det\left[ J_{v}(\phi(\mathbf{p}); S)\right]}.
		\end{equation}
		Observe that $\det\left[ J^{(i)}_{v}(\phi(\mathbf{p}); S)\right]=\sum_{i'} C_{i'i}(J_{v}(\phi(\mathbf{p}); S))$ and the right-hand side of (\ref{e:p:dd}) is $K_S$.
	
		Analogously, for the second derivative,
		$$D''_a=m_i \frac{\partial^2 q_i}{\partial p_a^2}, ~~ i\in S,$$
		which is implicitly given by $v_i(q(p_a,p_b^*))=\Delta p$. So for all $i\in S$,
		$$\sum_{i'\in S}\frac{\partial v_i}{\partial q_{i'}}\frac{\partial q_{i'}}{\partial p_a}=\mathbf{1}.$$
		Therefore, for all $i$
		$$\sum_{i'\in S}\left[ \frac{\partial}{\partial p_a}\left( \frac{\partial v_i}{\partial q_{i'}}\right) \frac{\partial q_{i'}}{\partial p_a}
		+\frac{\partial v_i}{\partial q_{i'}} \frac{\partial^2 q_{i'}}{\partial p_a^2}\right] =\mathbf{0}.$$
		As
		$$\frac{\partial}{\partial p_a}\left( \frac{\partial v_i}{\partial q_{i'}}\right) =\sum_{j\in S}\frac{\partial^2 v_i}{\partial q_{i'}\partial q_{j}}\frac{\partial q_{j}}{\partial p_a},$$
		we get
		$$\sum_{i'\in S}\left[ \left( \sum_{j\in S}\frac{\partial^2 v_i}{\partial q_{i'}\partial q_{j}}\frac{\partial q_{j}}{\partial p_a}\right)  \frac{\partial q_{i'}}{\partial p_a}
		+\frac{\partial v_i}{\partial q_{i'}} \frac{\partial^2 q_{i'}}{\partial p_a^2}\right] =\mathbf{0},$$
		and as such
		$$\sum_{i'\in S}\frac{\partial v_i}{\partial q_{i'}} \frac{\partial^2 q_{i'}}{\partial p_a^2}
		=
		-\sum_{i'\in S}\left[ \left( \sum_{j\in S}\frac{\partial^2 v_i}{\partial q_{i'}\partial q_{j}}\frac{\partial q_{j}}{\partial p_a}\right)  \frac{\partial 	q_{i'}}{\partial p_a}\right].$$
		The system on the left can be written as
		$$J_v(\boldsymbol{\sigma}; S) \nabla^2 q,$$
		and the system on the right can be written as
		$$\mathbf{k} H_{v_i}\mathbf{k}^{T},$$
		which, using the Cramer rule as before, leads to $R_S$ (instead of replacing a column by 1's we use the above).
	\end{proof}
	\bibliography{2021_rsoeiro-apinto_gne-bib}
\end{document}